\pdfoutput=1
\newif\ifFull
\Fulltrue

\ifFull
\documentclass[oribibl,11pt]{llncs}
\usepackage[margin=1in]{geometry}
\else
\documentclass[oribibl]{llncs}
\fi

\usepackage{amsmath,amssymb}
\usepackage{graphicx, hyperref}
\usepackage{cite}

\graphicspath{{figures/}}

\let\doendproof\endproof
\renewcommand\endproof{~\hfill\qed\doendproof}

\newcommand{\NP}{\mathsf{NP}}
\renewcommand{\P}{\mathsf{P}}

\DeclareMathOperator{\pagecross}{cr}

\DeclareMathOperator{\treewidth}{tw}
\DeclareMathOperator{\MSO}{MSO}
\DeclareMathOperator{\inc}{I}
\DeclareMathOperator{\minor}{\textsc{minor}}
\DeclareMathOperator{\outerplanar}{\textsc{outerplanar}}
\DeclareMathOperator{\onepage}{\textsc{onepage}}
\DeclareMathOperator{\twopage}{\textsc{twopage}}
\DeclareMathOperator{\separate}{separate}

\newcommand{\formula}[1]{\operatorname{\textsc{#1}}}

\usepackage{todonotes}

\title{Crossing Minimization for $1$-page and $2$-page Drawings of Graphs with Bounded Treewidth}
\author{Michael J. Bannister \and David Eppstein}
\institute{Department of Computer Science, University of California, Irvine}

\begin{document}
\maketitle
\begin{abstract}
We investigate crossing minimization for $1$-page and $2$-page book drawings. We show that computing the $1$-page crossing number is fixed-parameter tractable with respect to the number of crossings, that testing $2$-page planarity is fixed-parameter tractable with respect to treewidth, and that computing the $2$-page crossing number is fixed-parameter tractable with respect to the sum of the number of crossings and the treewidth of the input graph. We prove these results via Courcelle's theorem on the fixed-parameter tractability of properties expressible in monadic second order logic for graphs of bounded treewidth.
\end{abstract}

\pagestyle{plain} 

\section{Introduction}

A \emph{$k$-page book embedding} of a graph $G$ is a drawing that places the vertices of $G$ on a line (the \emph{spine} of the book) and draws each edge, without crossings, inside one of $k$ half-planes bounded by the line (the \emph{pages} of the book)~\cite{Kai-GC-1974,Oll-SCC-1973}. In one common drawing style, an \emph{arc diagram}, the edges in each page are drawn as circular arcs perpendicular to the spine~\cite{Wat-VIS-2002}, but the exact shape of the edges is unimportant for the existence of book embeddings. These embeddings can be generalized to \emph{$k$-page book drawings}: as before, we place each vertex on the spine and each edge within a single page, but with crossings allowed. The \emph{crossing number} of such a drawing is defined to be the sum of the numbers of crossings within each page, and the \emph{$k$-page crossing number} $\pagecross_k(G)$ is the minimum number of crossings in any $k$-page book drawing~\cite{ShaSykSze-GTCCS-1995}. In an optimal drawing, two edges in the same page cross if and only if their endpoints form interleaved intervals on the spine, so the problem of finding an optimal drawing may be solved by finding a permutation of the vertices and an assignment of edges to pages minimizing the number of pairs of edges with interleaved intervals on the same page.

As with most crossing minimization problems, $k$-page crossing minimization is $\NP$-hard; even the simple special case of testing whether the $2$-page crossing number is zero is $\NP$-complete~\cite{ChuLeiRos-SJADM-87}. However, it may still be possible to solve these problems in polynomial time for restricted families of graphs and restricted values of~$k$. For instance, recently Bannister, Eppstein and Simons~\cite{BanEppSim-GD-2013} showed the computation of $\pagecross_1(G)$ and $\pagecross_2(G)$ to be fixed-parameter tractable in the almost-tree parameter; here, a graph $G$ has almost-tree parameter $k$ if every biconnected component of $G$ can be reduced to a tree by removing at most $k$ edges.  In this paper we improve these results by finding fixed-parameter tractable algorithms for stronger parameters, allowing $k$-page crossing minimization to be performed in polynomial time for a much wider class of graphs.

\subsection{New results}
We design fixed-parameter algorithms for computing the minimum number of crossings $\pagecross_1(G)$ in a $1$-page drawing of a graph $G$, and the minimum number of crossings $\pagecross_2(G)$ in a $2$-page drawing of $G$. Ideally, fixed-parameter algorithms for crossing minimization should be parameterized by their \emph{natural parameter}, the optimal number of crossings. We achieve this ideal bound, for the first time, for $\pagecross_1(G)$. However, for $\pagecross_2(G)$, even testing whether a given graph is 2-page planar (that is, whether $\pagecross_2(G)=0$) is $\NP$-complete~\cite{ChuLeiRos-SJADM-87}. Therefore, unless $\P=\NP$, there can be no fixed-parameter-tractable algorithm parameterized by the crossing number. Instead, we show that $\pagecross_2(G)$ is fixed-parameter tractable in the sum of the natural parameter and the treewidth of $G$. One consequence of our result on $\pagecross_2(G)$ is that it is possible to test whether a given graph is 2-page planar, in time that is fixed-parameter tractable with respect to treewidth.

We construct these algorithms via Courcelle's theorem~\cite{Cou-IC-1990, Courcelle-Book}, which connects the expressibility of graph properties in monadic second order logic with the fixed-parameter tractability of these properties with respect to treewidth. Recall that second order logic extends first order logic by allowing the quantification of $k$-ary relations in addition to quantification over individual elements. In monadic second order logic we are restricted to quantification over unary relations (equivalently subsets) of vertices and edges. The property of having a 2-page book embedding is easy to express in (full) second-order logic, via the known characterization that a graph has such an embedding if and only if it is a subgraph of a Hamiltonian planar graph~\cite{BerKai-JCT-1979}. However, this expression is not allowed in monadic second-order logic because the extra edges needed to make the input graph Hamiltonian cannot be described by a subset of the existing vertices and edges of the graph. Instead, we prove a new structural description of $2$-page planarity that is more easily expressed in monadic second order logic.

\subsection{Related work}
As well as the previous work on crossing minimization for almost-trees~\cite{BanEppSim-GD-2013}, related results in fixed-parameter optimization of crossing number include a proof by Grohe, using Courcelle's theorem, that the topological crossing number of a graph is fixed-parameter tractable in its natural parameter~\cite{Gro-JCSS-2004}. This result was later improved by Kawarabayashi and Reed~\cite{KawRee-STOC-07}. Based on these results the crossing number itself was also shown to be fixed-parameter tractable; Pelsmajer et al. showed a similar result for the odd crossing number~\cite{PelSchSte-GD-07}.
In \emph{layered graph drawing}, Dujmovi\'c et al. showed that finding a drawing with $k$ crossings and $h$ layers is fixed-parameter tractable in the sum of these two parameters; this result depends on a bound on the pathwidth of such a drawing, a parameter closely related to its treewidth~\cite{Duj-A-2008}. 

Like many of these earlier algorithms, our algorithms have a high dependence on their parameter, rendering them impractical. For this reason we have not attempted an exact analysis of their complexity nor have we searched for optimizations to our logical formulae that would improve this complexity.

\section{Preliminaries}

\subsection{Bridges vs flaps and isthmuses}
There is an unfortunate terminological confusion in graph theory: two different concepts, a maximal subgraph that is internally connected by paths that avoid a given cycle, and an edge whose removal disconnects the graph,  are both commonly called \emph{bridges}.
We need both concepts in our algorithms. To avoid confusion, we call the subgraph-type bridges \emph{flaps} and the edge-type bridges \emph{isthmuses}. To be more precise, given a graph $G$ and a cycle $C$, we define an equivalence relation on the edges of $G\setminus C$ in which two edges are equivalent if they belong to a path that has no interior vertices in $C$, and we define a \emph{flap} of $C$ to be the subgraph formed by an equivalence class of this relation. (In general, different cycles will give rise to different flaps.) And given a graph $G$, we define an \emph{isthmus} of $G$ to be an edge of $G$ that does not belong to any simple cycles in $G$.

\subsection{Treewidth and graph minors}
The \emph{treewidth} of $G$ can be defined to be one less than the number of vertices in the largest clique in a chordal supergraph of $G$ that (among possible chordal supergraphs) is chosen to minimize this clique size~\cite{Bod-TCS-1998}. The problem of computing the treewidth of a general graph is $\NP$-hard~\cite{ArnCorPro-JADM-1987}, but it is fixed-parameter tractable in its natural parameter~\cite{Bod-STOC-1993}.

A graph $H$ is said to be a \emph{minor} of a graph $G$ if $H$ can be constructed from $G$ via a sequence edge contractions, edge deletions, and vertex deletions. It can be determined whether a graph $H$ is a minor of a graph $G$, in time that is polynomial in the size of $G$ and fixed-parameter tractable in the size of~$H$~\cite{RobSey-JCT-1995}.

\subsection{Logic of graphs}
We will be expressing graph properties in \emph{extended monadic second-order logic} ($\MSO_2$). This is a fragment of second-order logic that includes:
\begin{itemize}
\item
variables for vertices, sets of vertices, edges, and sets of edges;
\item
binary relations for equality ($=$), inclusion of an element in a set ($\in$) and edge-vertex incidence ($\inc$);
\item
the standard propositional logic operations: $\neg, \wedge, \vee, \to$;
\item
the universal quantifier ($\forall$) and the existential quantifier ($\exists$), both which may be applied to variables of any of the four variable types.
\end{itemize}
To distinguish the variables of different types, we will use $u,v,w,\ldots$ for vertices, $e,f,g,\ldots$ for edges, and capital letters for sets of vertices or edges (with context making clear which type of set).
Given a graph $G$ and an $\MSO_2$ formula~$\phi$ we write $G \models \phi$ (``$G$ models $\phi$'') to express the statement that $\phi$ is true for the vertices, edges, and sets of vertices and edges in $G$, with the semantics of this relation defined in the obvious way. $\MSO_2$ differs from full second order logic in that it allows quantification over sets, but not over higher order relations, such as sets of pairs of vertices that are not subsets of the given edges.
\ifFull
In Appendix~\ref{sec:mso2}, we provide a brief introduction to $\MSO_2$ logic in which we describe how to express some of the properties we need for our results.
\fi

The reason we care about expressing graph properties in $\MSO_2$ is the following powerful algorithmic meta-theorem due to Courcelle.

\begin{lemma}[Courcelle's theorem~\cite{Cou-IC-1990, Courcelle-Book}]
Given an integer $k \geq 0$ and an $\MSO_2$-formula $\phi$ of length $\ell$, an algorithm can be constructed that takes as input a graph $G$ of treewidth at most $k$ and decides in $O\big(f(k,\ell)\cdot(n+m)\big)$ time whether $G \models \phi$, where the function $f$ appearing in the time bound is a computable function of the treewidth $k$ and formula length~$\ell$.
\end{lemma}

\subsection{Combinatorial enumeration of crossing diagrams}
In order to show that the properties we study can be represented by logical formulas of finite length, we need to bound the number of combinatorially distinct ways that a subset of edges in a $k$-page graph drawing can cross each other.

We define a \emph{$1$-page crossing diagram} to be a placement of some points on the circumference of a circle, together with some straight line segments connecting the points such that each point is incident to a segment, no segment is uncrossed and no three segments cross at the same point. Two crossing diagrams are \emph{combinatorially equivalent} if they have the same numbers of points and line segments and there exists a cyclic-order-preserving bijection of their points that takes line segments to line segments.
The \emph{crossing number} of a $1$-page crossing diagram is the number of pairs of its line segments that cross each other.

We define a \emph{$2$-page crossing diagram} to be a $1$-page crossing diagram together with a labeling of its line segments by two colors. For a $2$-page crossing diagram we define the \emph{crossing number} to be the total number of crossing pairs of line segments that have the same color as each other.

\begin{lemma}\label{lem:crossing-diagrams}
There are $2^{O(k^2)}$ 1-page crossing diagrams with $k$ crossings, and there are $2^{O(k^2)}$ 2-page crossing diagrams with $k$ crossings.
\end{lemma}
\begin{proof}
Place $4k$ points around a circle. Then every $1$-page crossing diagram with $k$ or fewer crossings can be represented by choosing a subset of the points and a set of line segments connecting a subset of pairs of the points. There are $4k$ points and $4k(4k-1)/2$ pairs of points, so $2^{O(k^2)}$ possible subsets to choose.

Similarly, every $2$-page crossing diagram can be represented by a subset of the same $4k$ points, and two disjoint subsets of pairs of points, which again can be bounded by~$2^{O(k^2)}$.
\end{proof}

Two combinatorially equivalent crossing diagrams, as defined above, may have a topology that differs from each other, or from combinatorially equivalent diagrams with curved edges. This is because, for an edge with multiple crossings, the order of the crossings along this edge may differ from one diagram to another, but this ordering is not considered as part of the definition of combinatorial equivalence. For our purposes such differences are unimportant, as we are concerned only with the total number of crossings. So we consider two crossing diagrams to be equivalent if they have the same crossing pairs of edges, regardless of whether the crossings occur in the same order.

\section{$1$-page crossing minimization}

\subsection{Outerplanarity}
Recall that a graph is \emph{outerplanar} if there exists a placement of its vertices on the circumference of a circle such that when its edges are drawn as straight line segments they do not cross. Topologically, the circle and the half-plane are equivalent, so a graph is outerplanar if and only if it has a crossing-free $1$-page drawing.
For incorporating a test of outerplanarity into methods using Courcelle's theorem, it is convenient to use a standard characterization of the outerplanar graphs by forbidden minors:

\begin{lemma}[Chartrand and Harary~\cite{ChaHar-1967}]\label{lem:1-planar-minor}
A graph $G$ is outerplanar ($1$-page planar) if and only if it contains neither $K_4$ nor $K_{2,3}$ as a minor.
\end{lemma}

\begin{lemma}[Corollary 1.15 in \cite{Courcelle-Book}]
\label{lem:mso-minor}
Given any fixed graph $H$ there exists a $\MSO_2$-formula $\phi$ such that, for all graphs $G$, $G \models \phi$ if and only if $G$ contains $H$ as a minor.  We will write $\minor_H$ for $\phi$.
\end{lemma}

Let $\outerplanar$ be the formula $\neg \minor_{K_4} \wedge \neg \minor_{K_{2,3}}$. Then \autoref{lem:1-planar-minor} implies that, for all graphs $G$, $G \models \outerplanar$ if and only if $G$ is outerplanar. Because outerplanar graphs have bounded treewidth (at most two), Courcelle's theorem together with \autoref{lem:mso-minor} guarantee the existence of a linear time algorithm for testing outerplanarity. There are of course much simpler linear time algorithms for testing outerplanarity~\cite{Mit-IPL-1979,Wie-WG-1987}.

\subsection{Crossings vs treewidth}
Next, we relate the natural parameter for 1-page crossing minimization (the number of crossings) to the parameter for Courcelle's theorem (the treewidth). This relation will allow us to construct a fixed-parameter-tractable algorithm for the natural parameter.

\begin{figure}[t]
\centering
\includegraphics[scale=0.5]{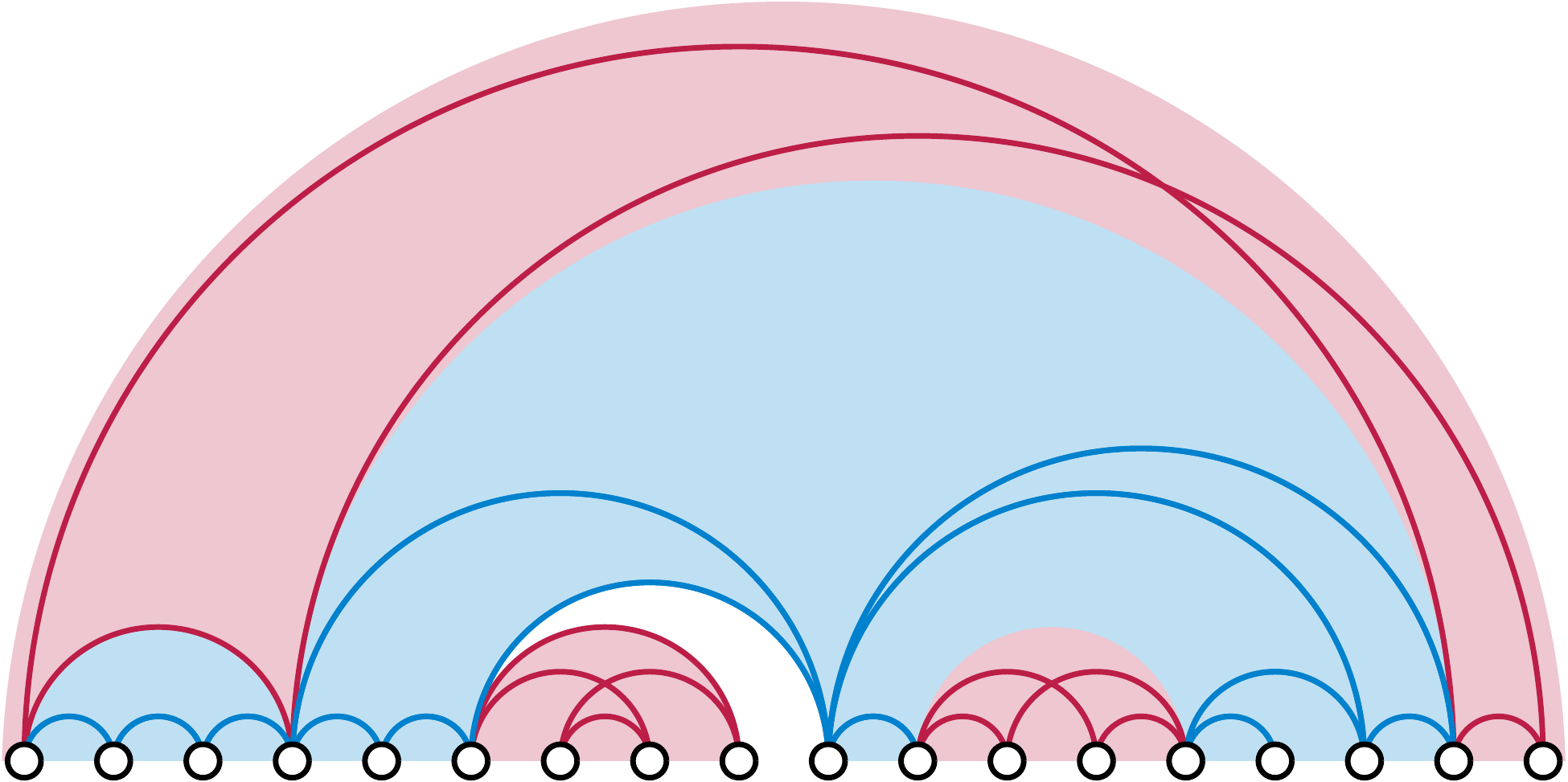}
\caption{An example of the clique-sum decomposition in \autoref{lem:1-page-crossing-treewidth}. The red regions represent the components with crossings and the blue regions represent outerplanar components. The entire graph may be reconstructed by performing clique-sums on the region boundaries.}
\label{fig:1-page-treewidth}
\end{figure}

A \emph{$k$-clique sum} of two disjoint graphs each containing a $k$-clique is formed by bijectively identifying each vertex of one $k$-clique with a vertex of the other $k$-clique, and then removing one or more of the $k$-clique edges from the resulting combined graph. 

\begin{lemma}[Lemma~1 in \cite{DEmHajMoh-AACO-02}]
If $G_1$ and $G_2$ each have treewidth at most $k$, then any clique-sum of $G_1$ and $G_2$ also has treewidth at most $k$.
\end{lemma}

\begin{lemma}\label{lem:1-page-crossing-treewidth}
Every graph $G$ has treewidth $O(\sqrt{\pagecross_1(G)})$.
\end{lemma}
\begin{proof}
Let $G$ be a graph with $\pagecross_1(G) = k$, and $D$ a $1$-page drawing of $G$ with $k$ crossings. Then let $H$ be the subgraph of $G$ induced by the endpoints of crossed edges in $D$. The remainder of $G$ after removing the edges in $H$ is a disjoint union of outerplanar graphs. Augment each connected component of $H$ and each outerplanar graph in the remainder of $G$ by adding edges between consecutive vertices along the spine of the drawing, completing a cycle around each connected component. From each augmented connected component $C$ we create a planar graph $C'$ by planarizing $C$ with respect to the drawing $D$. Since $C'$ is a planar graph with $O(k)$ vertices it has treewidth $O(\sqrt{k})$.  $C$ also has treewidth $O(\sqrt{k})$, as its treewidth is at most four times that of $C'$.

The graph $G$ may now be constructed from the augmented connected components and the outerplanar connected components by performing repeated $\{1,2\}$-clique-sums. Since each clique-sum preserves the treewidth, the graph $G$ has treewidth $O(\sqrt{k})$. An example of this construction is depicted in \autoref{fig:1-page-treewidth}.
\end{proof}

\subsection{Logical characterization}
Let $G$ be a graph with bounded $1$-page crossing number, and consider a drawing of $G$ achieving this crossing number. Then the set of crossing edges of the drawing partitions the halfplane into an arrangement of curves, and we can partition $G$ itself into the subgraphs that lie within each face of this arrangement. Each of these subgraphs is itself outerplanar, because it lies within a subset of the halfplane (with its vertices on the boundary of the subset) and has no more crossing edges; see \autoref{fig:1-page-crossing}. This intuitive idea forms the basis for the following characterization of the $1$-page crossing number, which we will use to construct an $\MSO_2$-formula for the property of having a drawing with low crossing number.

\begin{figure}[t]
\begin{minipage}[t]{0.35\linewidth}
\centering
\includegraphics[width=\textwidth]{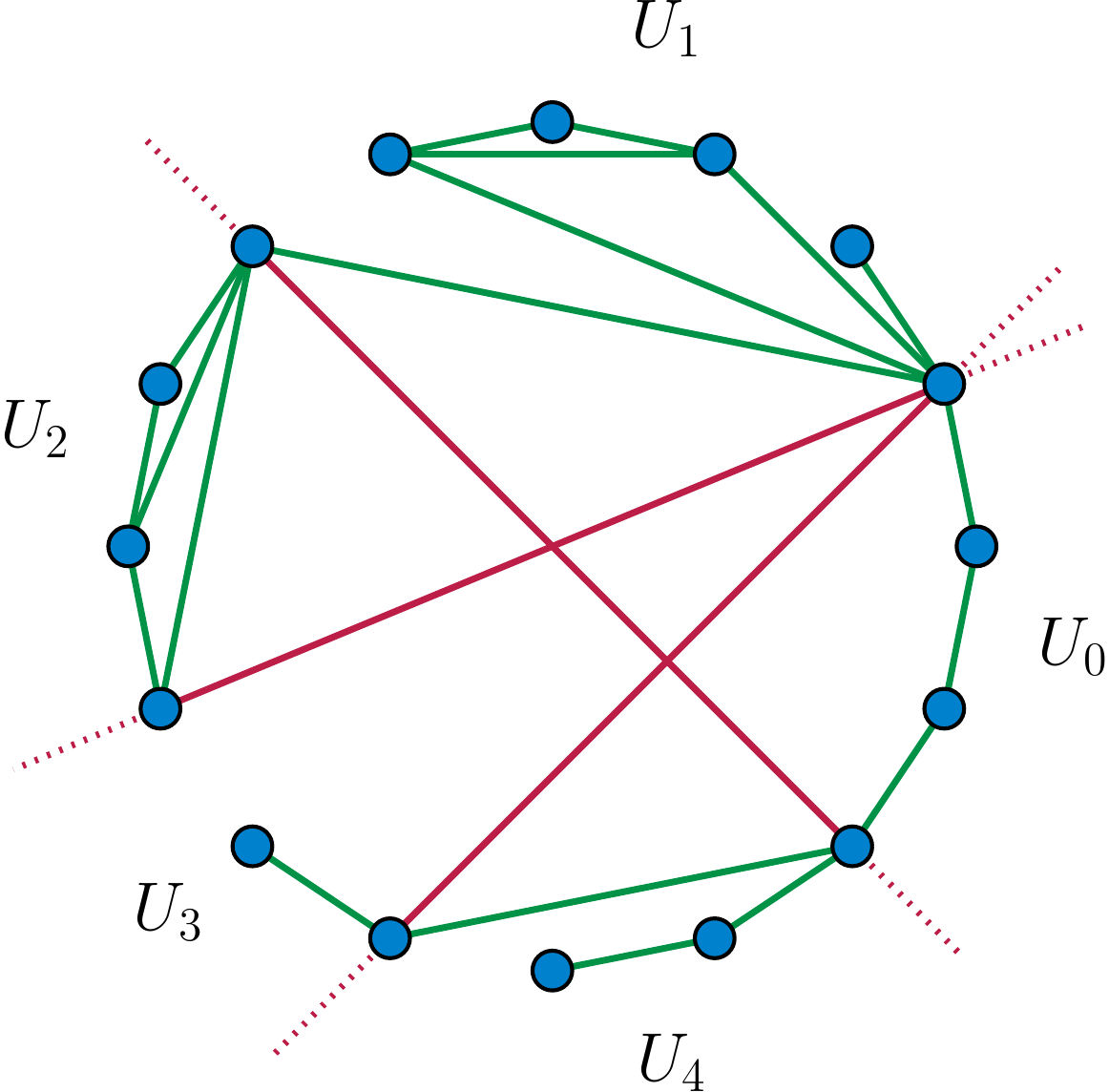}
\caption{A $1$-page drawing of a graph with two crossings and five outerplanar subgraphs.}
\label{fig:1-page-crossing}
\end{minipage}\hfill
\begin{minipage}[t]{0.57\linewidth}
\centering
\includegraphics[width=\textwidth]{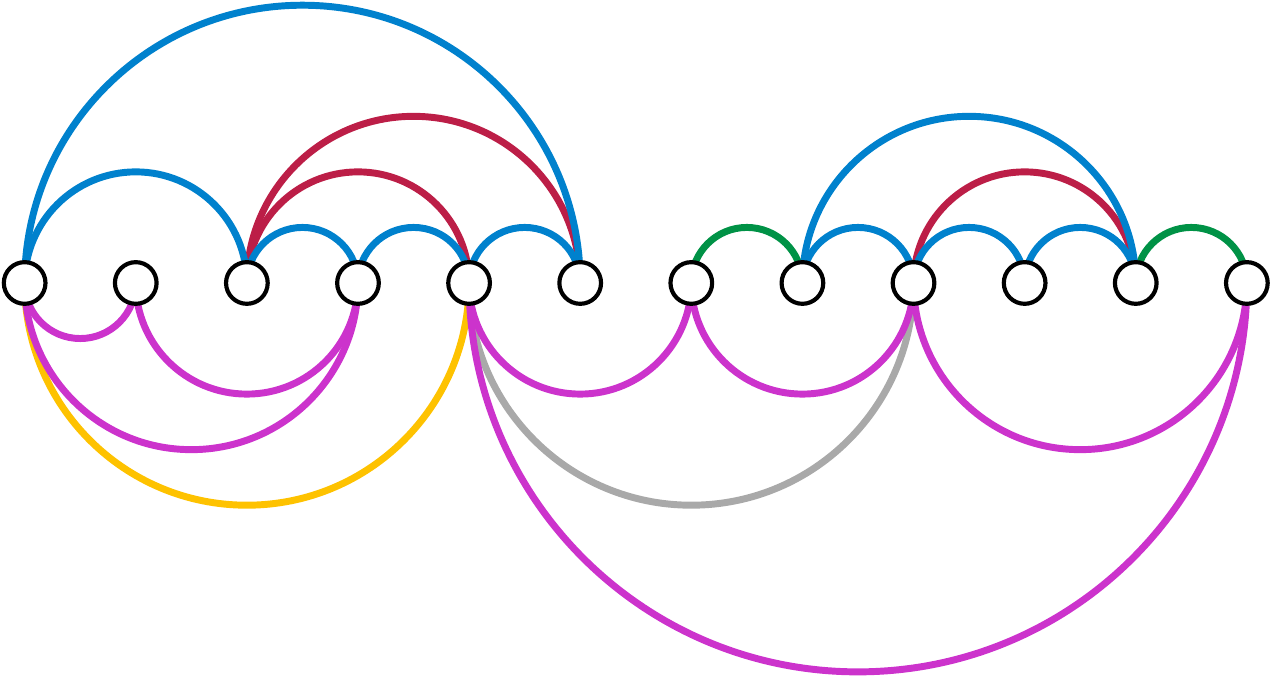}
\caption{A $2$-page planar graph with its edges partitioned into the six sets $A_b$ (green edges), $A_c$ (blue edges), $A_i$ (red edges), $B_b$ (yellow edges), $B_c$ (purple edges), and $B_i$ (gray edges).}
\label{fig:2-page-subsets}
\end{minipage}
\end{figure}

\begin{lemma}\label{lem:1-page-cross}
    A graph $G = (V,E)$ has $\pagecross_1(G)\le k$ if and only if there exist edges $F = \{e_0, \ldots, e_{r}\}$ with $r = O(k)$, vertices $W = \{v_0, \ldots, v_{\ell}\}$ with $\ell = O(k)$, and a partition $U_0, \ldots, U_{\ell}$ of $V \setminus W$ into (possibly empty) subsets, satisfying the following properties:
\begin{enumerate}
\item\label{cr1:endpoints} $W$ is the set of vertices incident to edges in $F$.
\item\label{cr1:induced} $F$ contains all edges in the induced subgraph on $W$.
\item\label{cr1:partition} There are no edges between $U_i$ and $U_j$ for $i\neq j$.
\item\label{cr1:outerplanar} There is an outerplanar embedding of the induced subgraph on $U_i \cup \{v_i, v_{i+1}\}$ with $v_i$ and $v_{i+1}$ adjacent for all $0 \leq i < \ell$.
\item\label{cr1:crossings} The edges in $F$ produce at most $k$ crossings when their endpoints (the vertices in~$W$) are placed in order according to their indices.
\end{enumerate}
\end{lemma}

We now construct a formula $\onepage_k$, based on \autoref{lem:1-page-cross}, such that $G \models \onepage_k$ if and only if $\pagecross_1(G)\le k$. The formula $\onepage_k$ will have the overall form of a disjunction, over all crossing configurations, of a conjunction of sub-formulas representing Properties \ref{cr1:endpoints}--\ref{cr1:outerplanar} in \autoref{lem:1-page-cross}. Property~\ref{cr1:crossings} will be represented implicitly, by the enumeration of crossing configurations.
The first three properties are easy to express directly: the formulas
\begin{align*}
\theta_1(W,F) &\equiv (\forall v)[v \in W \to (\exists e) [e \in F \wedge I(e,v)]]\\
\theta_2(F,W) &\equiv (\forall e)[(\forall v)[I(e,v) \to v \in W]\to e \in F] \\
\theta_3(U_i,U_j) &\equiv \neg (\exists e) (\exists u,v) [I(e,u) \wedge I(e,v) \wedge u \in U_i \wedge v \in U_j]
\end{align*}
express in $\MSO_2$ Properties \ref{cr1:endpoints}, \ref{cr1:induced}, and~\ref{cr1:partition} of \autoref{lem:1-page-cross} respectively.

To express Property~\ref{cr1:outerplanar} we first observe that it is equivalent to the property that the induced subgraph on $U_i \cup \{v_i,v_{i+1}\}$ with $v_i$ and $v_{i+1}$ identified (merged) to form a single supervertex is outerpalanar. That is, the requirement in Property~\ref{cr1:outerplanar} that vertices $v_i$ and $v_{i+1}$ be adjacent in the outerplanar embedding can be enforced by identifying the vertices. To express this property we need the following lemma, which can be proved in straightforward manner using the method of syntactic interpretations. (For details on this method see~\cite{EbbFluTho-1994,Gro-JCSS-2004}.)

\begin{lemma}\label{lem:mso-identify}
For every $\MSO_2$-formula $\phi$ there exists an $\MSO_2$-formula $\phi^*(v_1,v_2)$ such that $G \models \phi^*(a,b)$ if and only if $G/a\sim b \models \phi$, where $G/a\sim b$ is the graph constructed from $G$ by identifying vertices $a$ and $b$.
\end{lemma}

Now, to construct $\theta_4(U_i, v_i, v_j)$ we first modify the formula $\outerplanar$ by restricting its quantifiers to only quantify over vertices (and sets of vertices) in $U_i \cup \{v_i,v_j\}$ and edges (and sets of edges) between these vertices. This modified formula describes the outerplanarity of $U_i\cup\{v_i,v_j\}$. We then apply the transformation of \autoref{lem:mso-identify} to produce the formula $\theta_4(U_i, v_i, v_j)$, expressing the outerplanarity of the induced graph on $U_i \cup \{v_i, v_j\}$ with $v_i$ and $v_j$ identified.

\autoref{lem:crossing-diagrams} tells us that there are $2^{O(k^2)}$ ways of satisfying Property~\ref{cr1:crossings} of \autoref{lem:1-page-cross}. For each crossing diagram $D$ with $k$ crossings we can construct a formula $\alpha_D(v_0,\ldots, v_\ell, e_0,\ldots,e_r)$ specifying that the vertices $v_0, \ldots, v_\ell$ and edges $e_0,\ldots,e_r$ are in configuration $D$. We then construct the formula
\begin{multline*}
\beta_D \equiv (\exists v_0,\ldots v_\ell)(\exists e_0,\ldots,e_r)(\exists U_0, \ldots, U_\ell)\\
\Big[\alpha_D(v_0,\ldots, v_\ell, e_0,\ldots,e_r) \wedge \bigcup_0^\ell U_i = V \setminus \{v_0,\ldots,v_\ell\} \wedge \bigwedge_{i\neq j} U_i \cap U_j = \emptyset \\
\wedge \theta_1(v_0,\ldots,v_\ell; e_0,\ldots,e_r) \wedge \theta_2(e_0,\ldots,e_r; v_0,\ldots,v_\ell)
\wedge \bigwedge_{i\neq j} \theta_3(U_i,U_j)
\wedge \bigwedge_{i=0}^\ell \theta_4(U_i, v_i, v_{i+1}) \Big]
\end{multline*}
of length $O(k^2)$. This formula expresses the property that, in the given graph $G$, we can construct a crossing diagram of type $D$, and a corresponding partition of the vertices into subsets $U_i$, that obeys Properties \ref{cr1:endpoints}--\ref{cr1:outerplanar} of \autoref{lem:1-page-cross}. By \autoref{lem:1-page-cross}, this is equivalent to the property that $G$ has a $1$-page drawing with $k$ crossings in configuration $D$. Finally, we construct $\onepage_k$ by taking the disjunction of the $\beta_D$ where $D$ ranges over all crossing diagrams with $\le k$ crossings. Thus, $\onepage_k$ is a formula of length $2^{O(k^2)}$, expressing the property that $\pagecross_1(G)\le k$.

\begin{theorem}
There exists a computable function $f$ such that $\pagecross_1(G)$ can be computed in $O(f(k)n)$ time for a graph $G$ with $n$ vertices and with $k=\pagecross_1(G)$.
\end{theorem}
\begin{proof}
We have shown the existence of a formula $\onepage_k$ such that a graph $G \models \onepage_k$ if and only if $\pagecross_1(G)\le k$. By \autoref{lem:1-page-crossing-treewidth}, the treewidth of any graph with crossing number $k$ is $O(k)$. Applying Courcelle's theorem with the formula $\onepage_k$ and the $O(k)$ treewidth bound, it follows that computing $\pagecross_1(G)$ is fixed-parameter tractable in~$k$ .
\end{proof}

\section{$2$-page planarity}

A classical characterization of the graphs with planar $2$-page drawings is that they are exactly the subhamiltonian planar graphs:

\begin{lemma}[Bernhart and Kainen~\cite{BerKai-JCT-1979}]
\label{lem:subhamiltonian}
A graph is $2$-page planar if and only if it is the subgraph of planar Hamiltonian graph.
\end{lemma}

However, this characterization does not directly help us to construct an $\MSO_2$-formula expressing the $2$-page planarity of a graph, as we do not know how to construct a formula that asserts the existence of a supergraph with the given property. Hamiltonicity and planarity are both straightforward to express in $\MSO_2$, but there is no obvious way to describe a set of edges that may be of more than constant size, is not a subset of the existing edges, and can be used to augment the given graph to form a planar Hamiltonian graph.

For this reason we provide a new characterization, which we model on a standard characterization of planar graphs: a graph is planar if and only if, for every cycle $C$, the flaps of $C$ can be partitioned into two subsets (the interior and exterior of $C$) such that no two flaps in the same subset cross each other. For instance, this characterization has been used as the basis for a cubic-time divide and conquer algorithm for planarity testing, which recursively subdivides the graph into cycles and non-crossing subsets of flaps~\cite{AusPar-JMM-1961,Gol-GCC-1963,Shi-PhD-1969}. In our characterization of $2$-page graphs, we apply this idea to a special set of cycles, the boundaries of maximal regions within each halfplane that are separated from the spine of a 2-page book embedding by the edges of the embedding. The cycles of this type are edge-disjoint, and if a single cycle of this type has been identified then its interior flaps can also be identified easily: each interior flap is a single edge, and an edge forms an interior flap if and only if it belongs to the same page as the cycle in the book embedding and has both its endpoints on the cycle. As well as identifying which of the two pages each edge of a given graph is assigned to, our $\MSO_2$ formula will partition the edges into three different types of edge: the ones that belong to these special cycles, the ones that form interior flaps of these special cycles, and the remaining \emph{isthmus} edges that, if deleted, would disconnect parts of their page.

Suppose we are given a graph $G = (V,E)$ and a partition of its edges into two subsets $A,B$, intended to represent the two pages of a $2$-page drawing of~$G$. 
We define the graph $\separate(G;A,B)$ that splits each vertex of $G$ into two vertices, one in each page, with a new edge connecting them.
Thus, $\separate(G;A,B)$ has $2n$ vertices, which can be labeled by pairs of the form $(v,X)$ where $v$ is a vertex in $V$ and $X$ is one of the two sets in $A,B$. It has an edge between $(x,X)$ and $(y,Y)$ if either of two conditions is met: (1) $x=y$ and $X \neq Y$, or (2) $X=Y$ and there is an edge between $x$ and $y$ in $X$. See \autoref{fig:2-page-sep} for an illustration of the $\separate(G;A,B)$ construction.

\begin{lemma}\label{lem:2-page-planar}
A graph $G = (V,E)$ is $2$-page planar if and only if there exists a partition $A_b$, $A_c$, $A_i$, $B_b$, $B_c$, $B_i$ of $E$ into six subsets such that, for each of the two choices of $X=A$ and $X=B$, these subsets satisfy the following properties: 
\begin{enumerate}
\item\label{cr2:cycles} $X_c$ is a union of edge-disjoint cycles.
\item\label{cr2:isthmuses} $X_c\cup X_b$ does not contain any additional cycles that involve edges in $X_b$.
\item\label{cr2:inner} For every edge $e$ in $X_i$ there exists a cycle in $X_c$ containing both endpoints of~$e$.
\item\label{cr2:outerplanar} The graph formed by the edges $X_i \cup X_c \cup X_b$ is outerplanar.
\item\label{cr2:paths} For each cycle $C$ in $X_c$ it is not possible to find two vertex-disjoint paths $P_1$ and $P_2$ in $E$ such that neither path is a single edge in $X_i$, all four path endpoints are distinct vertices of $C$,
neither path contains a vertex of $C$ in its interior, and the two pairs of path endpoints are in crossing position on~$C$.
\item\label{cr2:planar} The subdivision $\separate(G;A_b \cup A_c \cup A_i, B_b \cup B_c \cup B_i)$ is planar.
\end{enumerate}
\end{lemma}
\newcommand{\lemtwopageplanar}{
Suppose $G$ has a $2$-page planar drawing. This drawing partitions the edges of $G$ into two sets $A$ and $B$. For $X = A$ or $B$, let $X_c$ be the set of edges $X$ forming a union of edge disjoint cycles that surround a maximal subset of their page. Then let $X_i$ be the edges in $X$ drawn in the interior of one of these cycles, and $X_b$ the remaining edges in $X$. It can be easily verified that the constructed partition satisfies Properties \ref{cr2:cycles} through~\ref{cr2:planar}.

Conversely, suppose we have a graph $G$ with a partition of its edges satisfying the properties of the lemma. By Property~\ref{cr2:planar}, $\separate(G;A_b \cup A_c \cup A_i, B_b \cup B_c \cup B_i)$ has a planar embedding. We may assume without loss of generality that, in this embedding, the cycles of $X_c$ given by Property~\ref{cr2:cycles}
separate the edges of $X_i$ (interior to the cycles) from the rest of the graph (exterior to the cycles).
For, by Property~\ref{cr2:outerplanar}, no two interior edges can cross, and by Property~\ref{cr2:paths}, no two exterior paths can cross. So, if we have a cycle in $X_c$ that does not properly separate $X_i$ from the rest of the graph, we may modify the embedding to flip the edges of $X_i$ into the interior of the cycle and to flip the components of the rest of the graph to the exterior of the cycle, preserving the (reflected) planar embedding of each flipped component, without introducing any new crossings. By performing this flipping operation to all cycles of $A_c$ and $B_c$, we obtain an embedding in which the cycles of $X_c$ separate $X_i$ from the rest of the graph, as stated above.

Next, given this embedding of $\separate(G;A_b \cup A_c \cup A_i, B_b \cup B_c \cup B_i)$, we contract all of the cycles ($X_c$) and isthmuses ($X_b$) in each page ($X = A$ and $B$), maintaining the orientation of the edges that were not contracted. As a consequence, the edges in $X_i$  within each cycle of $X_c$ are also contracted. However, in the embedding of $\separate(G;A_b \cup A_c \cup A_i, B_b \cup B_c \cup B_i)$, none of the contracted cycles surrounds any part of the graph that is not itself contracted. As a result, we are left with an embedding of a planar embedded bipartite multigraph that has one edge $(v,A)$--$(v,B)$ for each vertex $v$ in the original graph. Because this multigraph is bipartite, its dual graph has even degree at every vertex, and as the dual graph of a planar graph it is necessarily connected. Thus, the dual of the bipartite multigraph has an Euler tour, and (as with any Eulerian planar graph) this Euler tour can be made non-self-crossing by local uncrossing operations at each vertex. This tour can be represented geometrically as a Jordan curve $J$ that passes through the faces of the embedding of $\separate(G;A_b \cup A_c \cup A_i, B_b \cup B_c \cup B_i)$ (in some cases more than once per face) and crosses each edge $(v,A)$--$(v,B)$ exactly once.

From the embedding of $\separate(G;A_b \cup A_c \cup A_i, B_b \cup B_c \cup B_i)$ we can obtain a planar embedding of $G$ itself by contracting all the edges of the form $(v,A)$--$(v,B)$.
If we augment $G$ by adding an edge $uv$ between any two vertices $u$ and $v$ whose edges $(u,A)$--$(u,B)$ and $(v,A)$--$(v,B)$ are crossed consecutively by the Jordan curve $J$, then $J$ can be used to guide a non-crossing placement of these additional edges within the resulting embedding of~$G$. Thus, we have augmented $G$ to a Hamiltonian planar supergraph.
The Jordan curve passing through these contracted edges gives us a routing of a set of pairs

The planar dual of this graph has an Euler tour, as the primal graph is bipartite. This tour corresponds to Hamiltonian cycle in a planar supergraph of $G$, where edges are added between vertices if the edge does not already exist. The result that $G$ has a $2$-page book embedding follows by \autoref{lem:subhamiltonian}.
}


\begin{figure}[ht]
\centering
\includegraphics[width=0.5\textwidth]{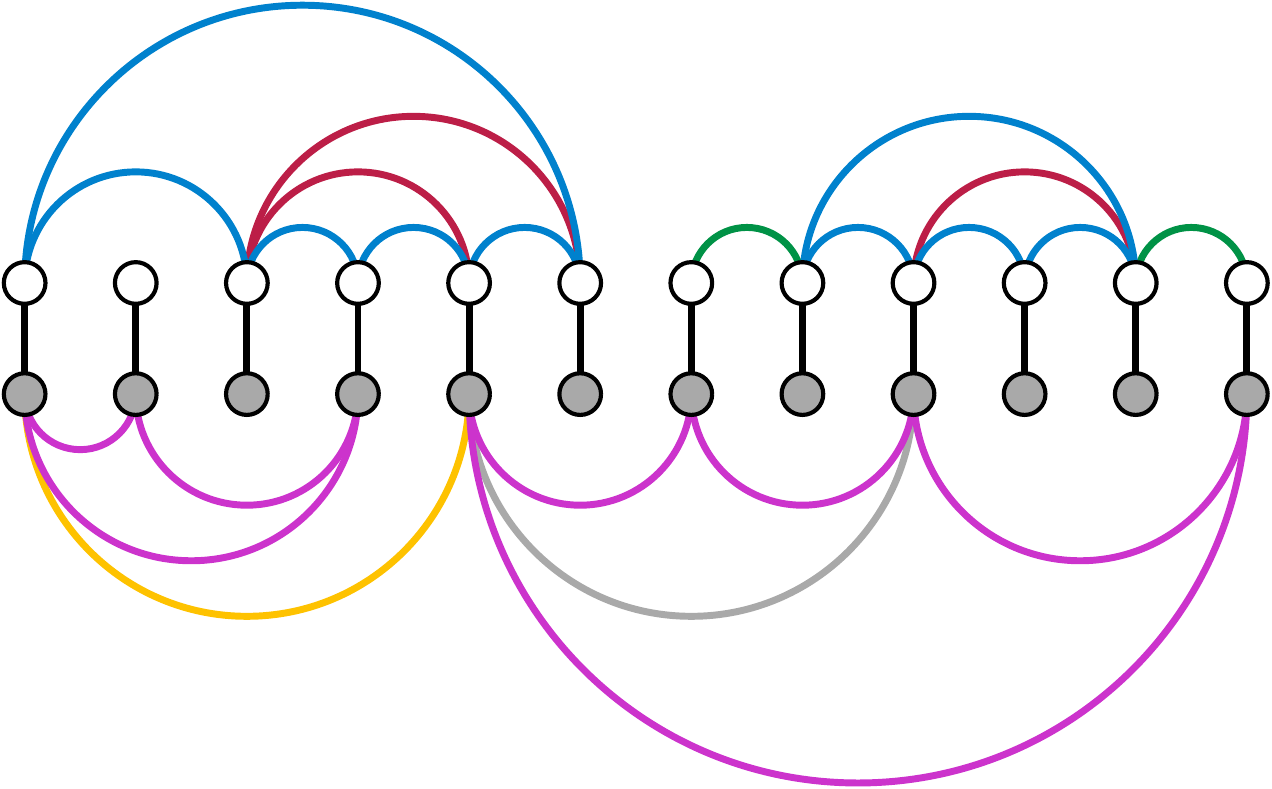}
\caption{The graph $\separate(G;A,B)$ where $G$ is the graph in \autoref{fig:2-page-subsets}, $A$ and $B$ are respectively the edges in the first and second page.}
\label{fig:2-page-sep}
\end{figure}

\autoref{fig:2-page-sep} illustrates the division of edge into six subsets described in \autoref{lem:2-page-planar}.
For the proof of  \autoref{lem:2-page-planar}, see
\ifFull
Appendix~\ref{sec:proof}.
\else
the full version of this paper.
\fi

We construct a formula $\twopage$ based on \autoref{lem:2-page-planar} with the property that $G \models \twopage$ if and only if $G$ is $2$-page planar. First, we construct formulas $\theta_1, \ldots, \theta_5$ expressing Properties \ref{cr2:cycles} through~\ref{cr2:paths} in \autoref{lem:2-page-planar}, as we did for $1$-page crossing; each of these properties has a straightforward expression in $\MSO_2$. To express Property~\ref{cr2:planar} we will need the following technical lemma, which can be proved using the method of syntactic interpretations.

\begin{lemma}\label{lem:mso-sep}
For every $\MSO_2$-formula $\phi$ there exists an $\MSO_2$-formula $\phi^*(A,B)$ such that $G \models \phi^*(A,B)$ if and only if $\separate(G;A,B) \models \phi$.
\end{lemma}

Now, we can express Property 6 as an $\MSO_2$-formula $\theta_6$ using \autoref{lem:mso-sep}, as planarity is expressible by \autoref{lem:mso-minor} and the fact that planar graphs are the graph that avoid $K_5$ and $K_{3,3}$ as minors. Thus, we define $\twopage$ to be the formula expressing the existence of $A_b, A_c, A_i, B_b,B_c,B_i$ satisfying $\theta_1, \ldots \theta_6$.

\begin{theorem}
There exists a computable function $f$ and an algorithm that can decide whether a given graph with treewidth $k$ is $2$-page planar in $O(f(k)n)$ time.
\end{theorem}
\begin{proof}
The result follows from Courcelle's theorem together with the construction of the $\MSO_2$ formula $\twopage$ representing the existence of a two-page planar embedding.
\end{proof}

\section{$2$-page crossing minimization}

We now extend the results of the previous section from $2$-page planarity to $2$-page crossing minimization. As in the $1$-page case, we will use a formula that involves a disjunction over crossing diagrams. Given a crossing diagram $D$ with $k$ crossings and $r+1$ edges, whose graph is $G$, we define the \emph{planarization} of $G$ with respect to $D$ to be the graph in which each edge $e_i$ is replaced by a path of degree four vertices, such that two of these replacement paths share a vertex if and only if the original two edges cross in $D$. As explained earlier, we do not care about the order of crossings along each edge (two crossing diagrams with the same sets of crossing pairs but with different crossing orders are considered equivalent. Nevertheless, we do preserve the order of crossings from (one representative of an equivalence class of) crossing diagrams to their planarizations, in order to ensure that the planarizations form planar graphs.

\begin{lemma}\label{lem:2-page-cross}
A graph $G = (V,E)$ has $\pagecross_2(G) = k$ if and only if there exists edges $e_0, e_1, \cdots, e_r$ with $r < 2k$ and a $2$-page crossing diagram $D$ with $k$ crossings on these edges such that when $G$ is planarized with respect to $D$ the resulting graph $G_D = (V_D, E_D)$ has a partition of $E_D$ into $A_b,A_c,A_i, B_b,B_c,B_i$ such that, for $X = A,B$:
\begin{enumerate}
\item $X_c$ is a union of edge disjoint cycles.
\item None of the cycles $X_c \cup X_b$ contains an edge in $X_b$.
\item If $e$ is an edge introduced in the planarization, then $e \in A_b \cup A_c \cup A_i$ if $e$ is in the first page of $D$, and $e \in B_b \cup B_c \cup B_i$ if it is in the second page of $D$.
\item\label{cr2:extended-inner} For every edge $e$ in $X_i$, there exists a subgraph $P$ containing $e$ and a cycle $C$ in $X_c$ such that $P$ consists only of vertices of $C$ and of degree-four vertices introduced in the planarization, $P$ contains at least two vertices of $C$, and $P$ includes all four edges incident to each of its planarization vertices.
\item For each two edges $e$ and $f$ in $X_i$, the two subgraphs $P_e$ and $P_f$ satisfying Property~\ref{cr2:extended-inner} do not each have a pair of endpoints in crossing position on the same cycle~$C$.
\item For each cycle $C$ in $X_c$ there do not exist two paths in $E$, such that neither path uses edges of $X_i$ or interior vertices of $C$, with four distinct endpoints on $C$ in crossing position.
\item the subdivision $\separate(G;A_b\cup A_c \cup A_i, B_b \cup B_c \cup B_i)$ is planar.
\end{enumerate}
\end{lemma}

Now, we construct a $\MSO_2$-formula $\zeta_k$ based on \autoref{lem:2-page-cross} such that $G \models \zeta_k$ if and only if $\pagecross_2(G) = k$. To handle the planarization process we use the following lemma. In the lemma, the notation $G^{e_1 \times e_2}$ describes  the graph obtained from a graph $G$ by deleting two edges $e_1$ and $e_2$ that do not share a common endpoint, and adding a new degree-4 vertex connected to the endpoints of $e_1$ and $e_2$.

\begin{lemma}[Grohe~\cite{Gro-JCSS-2004}]
    For every $\MSO_2$-formula $\phi$ there exists an $\MSO$-formula $\phi^*(x_1,x_2)$ such that $G \models \phi^*(e_1,e_2)$ if and only if $G^{e_1 \times e_2} \models \phi$.
\end{lemma}

Given any $\MSO_2$-formula $\phi$ and crossing diagram $D$, we can repeatedly apply the lemma above to construct a formula $\phi^D$ such that $G \models \phi^D(e_0,\ldots, e_r)$ if and only if $G_D \models \phi$. With this tool in hand it is straightforward to construct a formula $\gamma_D$ , expressing the property that, in a given graph $G$ we can build a crossing diagram with the structure of $D$, and partition the planarization $G_D$ into six sets, satisfying \autoref{lem:2-page-cross}. So we can define $\zeta_k$ to be the disjunction of the $\gamma_D$ ranging over all $2$-page crossing diagrams with $k$-crossings.

\begin{theorem}
There exists a computable function $f$ such that $\pagecross_2(G)$ can be computed in $O(f(k,t)n)$ time for a graph $G$ with $n$ vertices, $k = \pagecross_2(G)$, and $t = \treewidth(G)$.
\end{theorem}

\section{Conclusion}
We have provided new fixed-parameter algorithms for computing the crossing numbers for $1$-page and $2$-page drawings of graphs with bounded treewidth. The use of monadic second order logic and Courcelle's theorem in our solutions causes the running times of our algorithms to have an impractically high dependence on their parameters. We believe that it should be possible to achieve a better dependence by directly designing dynamic programming algorithms that use tree-decompositions of the given graphs, rather than by relying on Courcelle's theorem to prove the existence of these algorithms. Can this dependency be reduced to the point of producing practical algorithms? For $2$-page crossing minimization the runtime is parameterized by both the treewidth and the crossing number. Is $2$-page crossing minimization $\NP$-hard for graphs of fixed treewidth? We leave these questions open for future research.

It would also be of interest to determine whether three-page book embedding is fixed-parameter tractable in the treewidth or in other natural parameters of the input graphs. For this problem, we do not know of a logical characterization that would allow us to apply Courcelle's theorem. Even the special case of recognizing graphs with treewidth~$3$ that have three-page book embeddings would be of interest, to provide a computational attack on the still-open problem of whether there exist planar graphs that require four pages~\cite{DujWoo-DCG-07,Yan-STOC-86}.

\subsection*{Acknowledgments} This material is based upon work supported by the National Science Foundation under Grant CCF-1228639 and by the Office of Naval Research under Grant No. N00014-08-1-1015.

{\raggedright
\bibliographystyle{abuser}
\bibliography{paper}}

\ifFull
\newpage
\appendix

\section{Expressing graph properties in $\MSO_2$}
\label{sec:mso2}

For readers unfamiliar with $\MSO_2$ logic, we provide in this appendix some standard examples of graph properties that may be expressed in this logic, leading up to the properties that we use in our results. Additional examples may be found in the one of the standard introductions to graph logic~\cite{Courcelle-Book,FlumGrohe-Book, DowneyFellows-Book}. The building blocks in this section can be used to construct the formulas that we use throughout our paper.

Because the equal sign ($=$) is an element that is used within $\MSO_2$ formulas, expressing the equality relation between two vertices, edges, or sets, we instead use the equivalence sign ($\equiv$) to express the syntactic equality of two formulas, or the assignment of a name to a formula.

\subsection{$k$-Coloring}
The formula $\formula{color}_k$ that we construct below expresses the $k$-colorability of a graph. As a step towards the construction of $\formula{color}_k$, we first construct a formula $\formula{vertex-partition}$ expressing the property that a collection of vertex sets forms a partition of the vertices: the sets are disjoint from each other and their union contains all vertices in the graph.
\begin{multline*}
\formula{vertex-partition}(U_1, \ldots, U_k)
\equiv \\
(\forall v) \left[ \left( \bigvee_{i=1}^k v \in U_k \right) \wedge \left( \bigwedge_{i\neq j} \neg (v \in U_i \wedge v \in U_j) \right) \right] 
\end{multline*}
A formula $\formula{edge-partition}$ expressing the property that a collection of edge sets forms a partition of the edges in the graph may be constructed in the same way by changing vertex variables to edge variables and vertex set variables to edge set variables.

With the ability to partition vertices we can now construct $\formula{color}_k$. The construction uses the fact that a $k$-coloring forms a partition of the vertices with the additional property that, for every color class $C$, all edges have an endpoint of a different color than $C$.
\begin{multline*}
\formula{color}_k \equiv (\exists U_1, \ldots, U_k) \Big[ \formula{vertex-partition}(U_1, \ldots, U_k)\\ \wedge \bigwedge_{i=1}^k (\forall e)(\exists v)[\inc(e,v) \wedge v \not\in U_i] \Big]
\end{multline*}

\subsection{Minor containment and planarity}
Next, we construct a formula $\minor_H$ expressing the property that a graph has $H$ as a minor. If we label each of the $k$ vertices in $H$ with a distinct number in the range from $1$ to~$k$, then $H$ is a minor of $G$ if and only if there exists a corresponding collection of $k$ connected and disjoint subsets of the vertices of $G$, say $U_1$, \ldots, $U_k$, such that for each edge $(i,j)$ in $H$ there is an edge from $U_i$ to~$U_j$.

As part of this construction, we will use a formula $\formula{connected}$ expressing the property that a graph is connected. We will construct this formula by first constructing a formula $\formula{disconnected}$ expressing the property that a graph is disconnected. This is true if and only if the graph supports a nontrivial cut of the vertices with an empty cut-set.
\begin{multline*}
\formula{disconnected} \equiv
(\exists U)\Big[
(\exists u,v)\big[u \in U \wedge v \not\in U\big]\\
 \wedge \neg (\exists e)(\exists u,v)\big[\inc(e,u) \wedge \inc(e,v) \wedge u \in U \wedge v \not\in U\big]\Big]
\end{multline*}
We can now define $\formula{connected} \equiv \neg \formula{disconnected}$. A similar construction leads to formulas $\formula{connected-vertices}(V)$ and $\formula{connected-edges}(E)$ expressing the properties that vertex set $V$ describes a connected induced subgraph or that edge set $E$ describes a connected subgraph.

With the ability to express connectedness we can now construct $\minor_H$.
\begin{align*}
\formula{minor}_H
\equiv \exists(U_1, \ldots, U_k)\Bigg [&\bigwedge_{i=1}^k (\exists u)[u \in U_i]
\wedge
 \bigwedge_{i=1}^k \formula{connected-vertices}(U_i)\\
&\hspace{-0.85em}\wedge
\bigwedge_{i \neq j} (\forall v)[v \not\in U_i \vee v \not\in U_j] \\
&\hspace{-0.85em}\wedge\hspace{-0.95em}
\bigwedge_{(i,j) \in E_H} (\exists e) (\exists x,y)[\inc(e,x) \wedge \inc(e,y) \wedge x \in U_i \wedge y \in U_j] 
\Bigg ]
\end{align*}

Since the planar graphs are precisely the graphs that have neither $K_5$ nor $K_{3,3}$ as minors, we have
\[
\formula{planar} \equiv \neg \minor_{K_5} \wedge \neg \minor_{K_{3,3}}
\]
expressing the planarity of a graph in terms of these forbidden minors.

\subsection{Hamiltonicity}
Our last example will be a formula expressing the existence of a Hamiltonian cycle in a graph. A set of edges $F$ in a graph is a union of vertex-disjoint cycles if every edge in $F$ is adjacent to exactly two edges in $F$ other than itself. Thus,
\[
\formula{cycle-set}(F) \equiv (\forall e)\Big[e \in F \to (\exists^{=2} f)\big[f \in F \wedge e \neq f \wedge (\exists v)[\inc(e,v) \wedge \inc(f,v)]\big] \Big]
\]
expresses the property that $F$ is a disjoint union of cycles. (Here $\exists^{=2}$ is a logical shorthand for the existence of exactly two objects satisfying the given property, i.e. that there exist $f_1$ and $f_2$ both satisfying the property, that $f_1$ and $f_2$ are unequal, and that there do not exist three unequal edges all satisfying the property.) Then a set of edges is a single cycle if it is a union of cycles and forms a  connected subgraph. So we define
\[
\formula{cycle}(F) \equiv \formula{cycle-set}(F) \wedge \formula{connected-edges}(F),
\]
A set of edges $F$  spans a graph if every vertex is incident to at least one of the edges in~$F$.
\[
\formula{span}(F) \equiv (\forall v)(\exists e)[e \in F \wedge \inc(e, v)]
\]
Finally, a graph is Hamiltonian if it has a spanning cycle.
\[
\formula{hamiltonian} \equiv (\exists F)[\formula{cycle}(F) \wedge \formula{span}(F)]
\]

\section{Proof of \autoref{lem:2-page-planar}}
\begin{figure}
\centering
\includegraphics[width=0.45\textwidth]{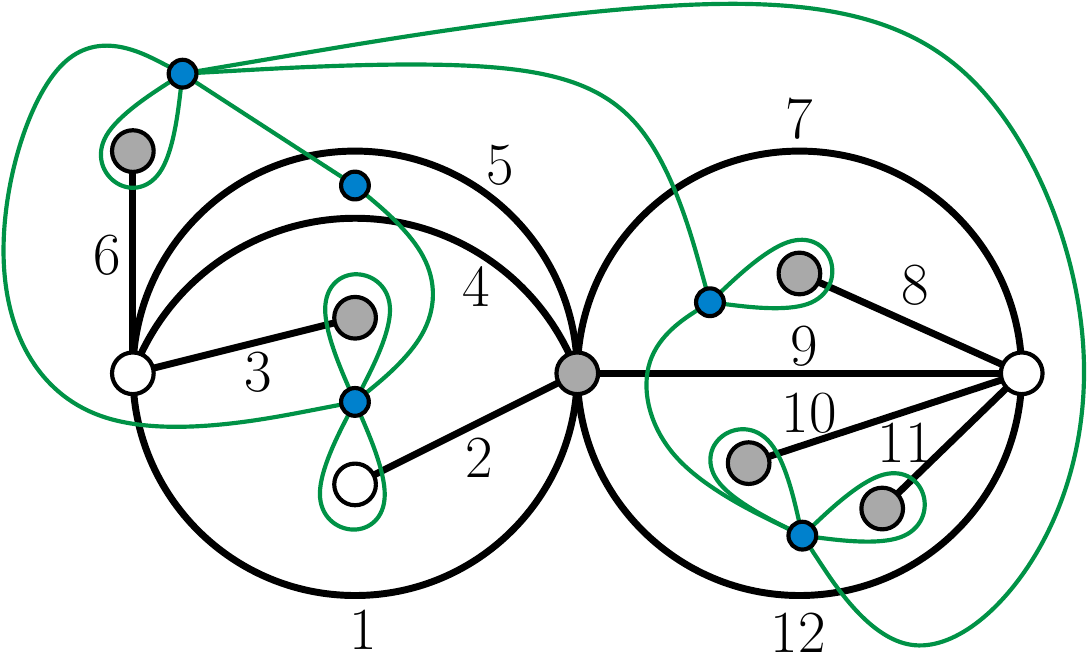}
\caption{The contraction of the graph in \autoref{fig:2-page-sep} and its planar dual (drawn with blue vertices and green edges). The edge labels correspond to the Hamiltonian cycle ordering of the vertices of $G$.}
\label{sec:proof}
\end{figure}
\lemtwopageplanar\qed

\fi
\end{document}